\documentclass[final,5p,times,twocolumn]{elsarticle}

\usepackage{amscd,bbm,amsmath,amsthm,graphicx,amssymb,latexsym,dsfont,amsfonts,paralist,epsfig,graphics,exscale,ulem}
\usepackage{mathrsfs}
\usepackage[colorlinks=true]{hyperref}
\usepackage{color,soul}

\newcommand{\bK}{\boldsymbol{K}}
\newcommand{\bsigma}{\boldsymbol{\sigma}}
\newcommand{\bi}{\boldsymbol{i}}
\newcommand{\bj}{\boldsymbol{j}}
\newcommand{\bx}{\boldsymbol{x}}

\newtheorem{theorem}{Theorem}[section]

\newtheorem{cor}[theorem]{Corollary}
\newtheorem{prop}[theorem]{Proposition}
\newtheorem{lem}[theorem]{Lemma}

\theoremstyle{definition}
\newtheorem{defn}[theorem]{Definition}
\newtheorem{example}[theorem]{Example}

\newtheorem{remark}[theorem]{Remark}

\numberwithin{equation}{section}

\newcommand{\EQ}{\begin{equation}\begin{array}{lllllllll}}
\newcommand{\EE}{\end{array}\end{equation}}
\newcommand{\MT}{\left[ \begin{array}{ccccccccc}}
\newcommand{\EM}{\end{array}\right]}

\newcommand{\eq}{\begin{equation}\begin{array}{lclllllllllllllll}}
\newcommand{\ee}{\end{array}\end{equation}}
\newcommand{\bmt}{\left[ \begin{array}{ccccccccc}}
\newcommand{\emt}{\end{array}\right]}
\newcommand{\bea}{\begin{eqnarray}}
\newcommand{\eea}{\end{eqnarray}}
\newcommand{\bean}{\begin{eqnarray*}}
\newcommand{\eean}{\end{eqnarray*}}

\journal{Automatica}
\begin{document}

\begin{frontmatter}


\title{Chaotic Characteristics of Discrete-time Linear Inclusion Dynamical Systems\tnoteref{label1}}

\tnotetext[label1]{Dai was supported by NSF of
China (Grant No.~11071112); T.~Huang was supported by National Priority Research Project NPRP 4-451-2-168 funded by Qatar Research Fund; Y.~Huang was supported partly by NSF of China (Grant No.~11071263) and the NSF of Guangdong Province; and Xiao in part by NSF 0605181 and 1021203 of the United States.}%

 \author[label2]{Xiongping Dai}
 \address[label2]{Department of Mathematics, Nanjing University, Nanjing 210093, People's Republic of China}
 \ead{xpdai@nju.edu.cn}

 \author[label3]{Tingwen Huang}
 \address[label3]{Texas A$\&$M University at Qatar, c/o Qatar Foundation, P.O. Box 5825, Doha, Qatar}
 \ead{tingwen.huang@qatar.tamu.edu}

 \author[label4]{Yu Huang}
 \address[label4]{Department of Mathematics, Zhongshan (Sun Yat-Sen) University, Guangzhou 510275, People's Republic of China}
 \ead{stshyu@mail.sysu.edu.cn}

 \author[label5]{Mingqing Xiao}
\address[label5]{Department of Mathematics, Southern Illinois University, Carbondale, IL 62901-4408, USA}
 \ead{mxiao@math.siu.edu}

\begin{abstract}
Given $K$ real $d$-by-$d$ nonsingular matrices $S_1,\dotsc,S_K$, by extending the well-known Li-Yorke chaotic description of a deterministic nonlinear dynamical system, to a discrete-time linear inclusion/control dynamical system
\begin{equation*}
x_n\in\left\{S_1, \dotsc,S_K\right\}x_{n-1},\quad x_0\in\mathbb{R}^d\textrm{ and }n\ge1,
\end{equation*}
we study the irregularity of orbit $(x_n(x_0,\sigma))_{n\ge1}$, governed by the law $\sigma\colon\mathbb{N}\rightarrow\{1,\dotsc,K\}$, for any initial state $x_0\in\mathbb{R}^d$. A sufficient condition is given so that for a \textit{residual} subset of the space
of all possible switching laws $\sigma$, we have
\begin{equation*}
\begin{cases}\liminf\limits_{n\to+\infty}\|x_n(x_0,\sigma)-x_n(y_0,\sigma)\|=0\\
\limsup\limits_{n\to+\infty}\|x_n(x_0,\sigma)-x_n(y_0,\sigma)\|=+\infty
\end{cases}\forall x_0,y_0\in\mathbb{R}^d \textrm{with }x_0\not=y_0.
\end{equation*}
We also show that a periodic stable inclusion system will not possess any
such irregular states.
\end{abstract}
\end{frontmatter}

\section{Introduction}\label{sec1}
Chaotic behavior is an important subject in study of the theory of dynamical systems. This type of systems is highly sensitive to initial conditions, and small perturbations in initial conditions (such as those due to rounding errors in numerical computation) yield widely diverging outcomes, rendering long-term prediction impossible in general. Even if a system is deterministic, i.e. their future behavior is fully determined by their initial conditions with no random elements involved, the long-term prediction of its chaotic behavior is still impossible.
In this paper we employ the idea of Li-Yorke to study the irregular behavior of a discrete-time linear
inclusion/control dynamical system.

\subsection{Basic concept}\label{sec1.1}
Let $\bK=\{1,\dotsc,K\}$ endowed with the discrete topology and let $S_1,\dotsc,S_K$ be $K$ nonsingular real $d\times d$ matrices, where $K\ge2$ and $d\ge2$.
This then induces a discrete-time linear
inclusion/control dynamical system:
\begin{equation}\label{eq1.1}
x_n\in\left\{S_kx_{n-1}\right\}_{k\in\bK},\quad x_0\in\mathbb{R}^d\textrm{ and }n\ge1,
\end{equation}
where $x_0$ is the initial state. Write
\begin{subequations}
\begin{gather}
\varSigma_{\bK}^+=\{\sigma\colon\mathbb{N}\rightarrow\bK\},\quad \textrm{where }\mathbb{N}=\{1,2,\dotsc\},\label{eq1.2a}\\
\intertext{which is equipped with the standard compact product topology compatible with the metric given by}
d(\sigma,\sigma^\prime)=\sum_{n=1}^{+\infty}\frac{\min\left\{1,|\sigma(n)-\sigma^\prime(n)|\right\}}{2^{n}}\quad\forall \sigma,\sigma^\prime\in\varSigma_{\bK}^+.\label{eq1.2b}
\end{gather}
\end{subequations}
Then for any $\sigma\in\varSigma_{\bK}^+$, to any initial state $x_0\in\mathbb{R}^d$ the corresponding output $(x_n(x_0,\sigma))_{n\ge1}$ of System (\ref{eq1.1}), governed by $\sigma$, is defined as
$x_n(x_0,\sigma)=S_{\sigma(n)}x_{n-1}$ for all $n\ge1$.

System $(\ref{eq1.1})$ has recently been found in many real applications. For the theoretic and applied importance of the study of System $\mathrm{(\ref{eq1.1})}$, readers may see, e.g., \cite{Lib, SG}.

Recall that a subset of a complete metric space is said to be \textit{residual} if it contains a dense $G_\delta$-set. So a residual subset is very large from the point of view of topology.

To describe the complexity of the dynamics of the output $(x_n(x_0,\sigma))_{n\ge1}$ of System (\ref{eq1.1}) as time evolves, we now introduce the dynamical concept---chaos, which is motivated by the sensitive dependence on initial conditions in Li-Yorke's definition of chaos \cite{LY} for nonlinear dynamical systems.

\begin{defn}\label{def1.1}
A switching law $\sigma\in\varSigma_{\bK}^+$ is said to be \textit{chaotic} for System $(\ref{eq1.1})$ if for all $x_0\in\mathbb{R}^d\setminus\{0\}$
\begin{equation*}
\liminf_{n\to+\infty}\|x_n(x_0,\sigma)\|=0\quad \textrm{and}\quad\limsup_{n\to+\infty}\|x_n(x_0,\sigma)\|=+\infty.
\end{equation*}
System $(\ref{eq1.1})$ is called \textit{chaotic} if its chaotic switching laws form a residual subset of the space $\varSigma_{\bK}^+$.
\end{defn}

Since $x_n(x_0,\sigma)-x_n(y_0,\sigma)=x_n(x_0-y_0,\sigma)$ for any initial states $x_0$ and $y_0$ in $\mathbb{R}^d$, it is easily seen that a switching law $\sigma\in\varSigma_{\bK}^+$ is  chaotic for System $(\ref{eq1.1})$ if and only if
\begin{subequations}
\begin{align}
&\liminf_{n\to+\infty}\|x_n(x_0,\sigma)-x_n(y_0,\sigma)\|=0 \label{eq1.3a}\\
\intertext{and}
&\limsup_{n\to+\infty}\|x_n(x_0,\sigma)-x_n(y_0,\sigma)\|=+\infty. \label{eq1.3b}
\end{align}\end{subequations}
for all $x_0,y_0\in\mathbb{R}^d$ with $x_0\not=y_0$.

In topological dynamical system, $\mathrm{(\ref{eq1.3a})}$ and $\mathrm{(\ref{eq1.3b})}$ are respectively called the proximal and distal properties. However, our distal property $\mathrm{(\ref{eq1.3b})}$ is much more stronger than the general Li-Yorke's one~\cite{LY} that only requires
\begin{equation*}
\limsup_{n\to+\infty}\|x_n(x_0,\sigma)-x_n(y_0,\sigma)\|>0.
\end{equation*}
This sensitivity means that any two trajectories governed by the same chaotic switching law $\sigma$ will be bound to get close together for a while, as time evolves, and then to go far away from each other for a while, and such dynamics will be repeated infinitely that leads to irregular, complex dynamical behaviors.

We note here that Balde and Jouan introduced in \cite{BJ} a kind of chaotic switching laws. However, these two kinds of definitions are essentially different. Balde-Jouan's is completely based on the topological structure of a switching law $\sigma\in\varSigma_{\bK}^+$; but ours is one having to do with the stability and instability of System (\ref{eq1.1}) rather than the single topological structure of the switching law $\sigma$. See Section~\ref{sec2} for the details.

\subsection{Main statement}\label{sec1.2}
In this paper we present, for System $\mathrm{(\ref{eq1.1})}$, a simple mechanism of generating the chaotic dynamics described as in Definition~\ref{def1.1}, as follows:

\begin{theorem}\label{thm1.2}
System $(\ref{eq1.1})$ is chaotic in the sense of Definition~\ref{def1.1}, if there are two words
$(i_1,\dotsc,i_m)\in\bK^m$ and $(j_1,\dotsc,j_n)\in\bK^n$
such that
$\|S_{i_m}\dotsm S_{i_1}\|<1<\|S_{j_n}\dotsm S_{j_1}\|_{\mathrm{co}}$.
\end{theorem}

Here $\|\cdot\|$ and $\|\cdot\|_{\mathrm{co}}$ denote the usual matrix maximum norm and minimum norm, respectively, defined by
\begin{equation*}
\|A\|=\max_{x\in\mathbb{R}^d,\|x\|=1}\|Ax\|\quad \textrm{and}\quad\|A\|_{\textrm{co}}=\min_{x\in\mathbb{R}^d,\|x\|=1}\|Ax\|
\end{equation*}
for any $A\in\mathrm{GL}(d,\mathbb{R})$.
\subsection{Outline}
This note is simply organized as follows: In Section~\ref{sec2} we shall study the topological structure of a nonchaotic switching law.
We will prove our main result Theorem~\ref{thm1.2} in Section~\ref{sec3}.
Finally in Section~\ref{sec4}, we will show that every periodically stable inclusion system does not have any chaotic behaviors (Corollary~\ref{cor4.2}). So, a periodically stable inclusion system is ``simple'' from our viewpoint of chaos.

\section{Chaotic switching laws}\label{sec2}

This section is devoted to comparing our definition of chaotic switching law with that of Balde and Jouan introduced in \cite{BJ}. In addition, we shall present some criteria for nonchaotic dynamics in our sense of Definition~\ref{def1.1}.

\subsection{Balde and Jouan's definition of chaos}\label{sec2.1}
Let $\{S_1,\dotsc,S_K\}\subset\mathbb{R}^{d\times d}$, not necessarily nonsingular, and then we consider the induced linear inclusion system
\begin{equation}\label{eq2.1}
x_n\in\left\{S_1,\dotsc,S_K\right\}x_{n-1},\quad x_0\in\mathbb{R}^d\textrm{ and }n\ge1.
\end{equation}
For an arbitrary matrix $A\in\mathbb{R}^{d\times d}$, let $\lambda_1,\dotsc,\lambda_\kappa$ be its all distinct eigenvalues. Then we write
\begin{equation*}
\rho_{\!A}^{}=\max\{|\lambda_1|,\dotsc,|\lambda_\kappa|\},
\end{equation*}
which is called the \textit{spectral radius} of $A$.

A recent definition of chaotic switching law has been given by Balde and Jouan \cite{BJ} as follows:

\begin{defn}[\cite{BJ}]\label{def2.1}
A switching law $\sigma=(\sigma(n))_{n\ge1}\in\varSigma_{\bK}^+$ is called {\it nonchaotic} in the sense of Balde and Jouan, if to any sequence $\langle n_i\rangle_{i\ge1}\nearrow+\infty$ and any $m\ge1$ there corresponds some $\delta$ with $2\le\delta\le m+1$ such that for all $\ell\ge1$, there exists an $\ell_0\ge \ell$ so that $\sigma$ is constant restricted to some subinterval of $[n_{\ell_0}, n_{\ell_0}+m]$ of length greater than or equal to $\delta$.
\end{defn}

Clearly, a constant switching law $\sigma$ with $\sigma(n)\equiv k$ for all $n\ge1$, for some $1\le k\le K$, is nonchaotic in the sense of Balde and Jouan; meanwhile, it is also nonchaotic in the sense of our Definition~\ref{def1.1}. In fact, we can obtain a more general result.

\begin{prop}\label{prop2.2}
If $\sigma\in\varSigma_{\bK}^+$ is a periodic switching law, then it is nonchaotic for System $(\ref{eq2.1})$ in the sense of Definition~\ref{def1.1}.
\end{prop}

\begin{proof}
Since $\sigma$ is periodical, it can be written as
\begin{equation*}
\sigma=(\uwave{k_1,\dotsc,k_\pi},\uwave{k_1,\dotsc,k_\pi},\dotsc)
\end{equation*}
for some word $(k_1,\dotsc,k_\pi)\in\bK^\pi$ of length $\pi\ge1$. Simply set $A=S_{k_\pi}\dotsm S_{k_1}$. If the spectral radius $\rho_{\!A}^{}$ of $A$ is less than $1$, then from the classical Gel'fand spectral-radius formula
\begin{equation*}
\lim_{n\to+\infty}\frac{1}{n}\log\|S_{\sigma(n)}\dotsm S_{\sigma(1)}\|=\frac{1}{\pi}\log\rho_{\!A}^{}<0.
\end{equation*}
So,
\begin{equation*}
\lim_{n\to+\infty}\|x_n(x_0,\sigma)\|=0\quad\forall x_0\in\mathbb{R}^d\setminus\{0\},
\end{equation*}
which means that $\sigma$ is nonchaotic for System (\ref{eq2.1}) in the sense of Definition~\ref{def1.1}. If $\rho_{\!A}^{}\ge1$, then one can find a unit vector $\bx_0\in\mathbb{R}^d$ and an eigenvalue $\lambda$ of $A$ with $|\lambda|\ge1$ such that
$$
A^n\bx_0=\lambda^n\bx_0\quad\forall n\ge1,
$$
which implies that
$$
\liminf_{n\to+\infty}\|x_n(\bx_0,\sigma)\|>0,
$$
and so $\sigma$ is nonchaotic for System (\ref{eq2.1}) in the sense of Definition~\ref{def1.1}.

This concludes the statement of Proposition~\ref{prop2.2}.
\end{proof}

Balde and Jouan's definition~\ref{def2.1} of chaos only depends on the single switching law $\sigma$ and ignores the structure of System $(\ref{eq1.1})$ or $(\ref{eq2.1})$, which is not enough to capture the essential of chaos of System $(\ref{eq1.1})$. The following lemma gives the key property of a Balde-Jouan nonchaotic switching law.

\begin{lem}\label{lem2.3}
Let $\sigma\in\varSigma_{\bK}^+$ be a nonchaotic switching law in the sense of Balde and Jouan. Then, there exists some alphabet $k\in\{1,\dotsc,K\}$ such that for any $\ell\ge1$ and any $\ell^\prime\ge1$, there exists an $n_\ell\ge\ell^\prime$ so that $\sigma(n_\ell+1)=\dotsm=\sigma(n_\ell+\ell)=k$.
\end{lem}

\begin{proof}
First, for the nonchaotic $\sigma$ we can choose a sequence $\langle n_i\rangle_{i\ge1}\nearrow+\infty$ and some $k\in\{1,\dotsc,K\}$, which are such that $n_{i+1}-n_i\nearrow+\infty$ and $\sigma(n_i)=k$ for all $i\ge1$. Now from Definition~\ref{def2.1} with $m=1$, it follows that we can choose a subsequence of $\langle n_i\rangle_{i\ge1}$, still write, without loss of generality, as $\langle n_i\rangle_{i\ge1}$, such that $\sigma(n_i)=\sigma(n_i+1)=k$ for all $i\ge1$. Repeating this procedure for $\langle n_{i}+1\rangle_{i\ge1}$ proves the statement of Lemma~\ref{lem2.3}.
\end{proof}

However, our chaotic property is a kind of dynamical behavior, which discovers the complexity of the structure of the outputs of the inclusion/control system $\mathrm{(\ref{eq1.1})}$ or $(\ref{eq2.1})$, as shown by Lemma~\ref{lem2.4} below. And from Proposition~\ref{prop2.2}, it also depends on the topological structure of the switching law $\sigma$ itself.

\begin{lem}\label{lem2.4}
Let System $(\ref{eq2.1})$ be defined by
\begin{equation*}
S_1=\left(\begin{matrix}2^{-1}&0\\0&2^{-1}\end{matrix}\right)\quad \textrm{and}\quad S_2=\left(\begin{matrix}2&0\\0&2\end{matrix}\right).
\end{equation*}
Then for System $(\ref{eq2.1})$, the switching law $\pmb{\sigma}\in\varSigma_2^+$ given as
\begin{equation*}
(11,2222,\stackrel{2^{3}\textrm{-folds}}{\overbrace{1\dotsm1}},\stackrel{2^{4}\textrm{-folds}}{\overbrace{2\dotsm2}},\dotsc,\stackrel{2^{2n-1}\textrm{-folds}}{\overbrace{1\dotsm 1}},\stackrel{2^{2n}\textrm{-folds}}{\overbrace{2\dotsm 2}},\dotsc)
\end{equation*}
is chaotic under the sense of Definition~\ref{def1.1}, but $\pmb{\sigma}$ is nonchaotic in the sense of Balde and Jouan.
\end{lem}

\begin{proof}
The statement comes easily from Definitions~\ref{def1.1} and \ref{def2.1} and we thus omit the details here.
\end{proof}

In fact, we can show this system is chaotic under the sense of Definition~\ref{def1.1} from Theorem~\ref{thm1.2}.

\subsection{An ergodic-theoretic viewpoint}\label{sec2.2}
Next, we will study a case where the chaotic behavior does not occur from the ergodic-theoretic viewpoint. Let
\begin{equation}\label{eq2.2}
\theta\colon\varSigma_{\bK}^+\rightarrow\varSigma_{\bK}^+;\quad \sigma=(\sigma(n))_{n\ge1}\mapsto\theta(\sigma)=(\sigma(n+1))_{n\ge1}
\end{equation}
be the one-sided shift transformation on the compact metrizable space $\varSigma_{\bK}^+$ of all the possible switching laws of System $\mathrm{(\ref{eq2.1})}$ as in Section~\ref{sec1}.

Recall that a probability measure $\mu$ on the Borel measurable space $(\varSigma_{\bK}^+,\mathscr{B}(\varSigma_{\bK}^+))$ is \textit{invariant} if $\mu(B)=\mu(\theta^{-1}B)$ for all $B\in\mathscr{B}(\varSigma_{\bK}^+)$; further an invariant probability measure $\mu$ is called \textit{ergodic} if $\mu(B)=0$ or $1$ whenever $\mu(B\vartriangle\theta^{-1}B)=0$, where $\vartriangle$ stands for the symmetric difference of two subsets.

For System (\ref{eq2.1}), it is very convenient to consider the corresponding linear cocycle
\begin{equation}\label{eq2.3}
\mathcal {S}\colon\mathbb{N}\times\varSigma_{\bK}^+\rightarrow\mathbb{R}^{d\times d};\quad (n,\sigma)\mapsto\mathcal {S}(n,\sigma)=S_{\sigma(n)}\dotsm S_{\sigma(1)}
\end{equation}
driven by the one-sided shift transformation $\theta$. According to the Oselede\v{c} multiplicative ergodic theorem~\cite{Ose}, we can obtain the following result.

\begin{prop}\label{prop2.5}
Let $\mu$ be an ergodic probability measure of the one-sided shift $\theta$ on $\varSigma_{\bK}^+$. If $\mathcal {S}$ has either a positive Lyapunov exponent or a negative Lyapunov exponent at $\mu$, then for $\mu$-a.e. $\sigma\in\varSigma_{\bK}^+$ it is nonchaotic for System $(\ref{eq2.1})$ in the sense of Definition~\ref{def1.1}.
\end{prop}

\begin{proof}
Let $\lambda<0$ be a Lyapunov exponent of $\mathcal {S}$ at $\mu$. Then from the Oselede\v{c} multiplicative ergodic theorem~\cite{Ose}, it follows that for $\mu$-a.e. $\sigma\in\varSigma_{\bK}^+$ there exists a corresponding unit vector, say $\bx_0(\sigma)\in\mathbb{R}^d$, such that
$$
\lim_{n\to+\infty}\frac{1}{n}\log\|\mathcal {S}(n,\sigma)\bx_0(\sigma)\|=\lambda.
$$
So,
$$
\limsup_{n\to+\infty}\|x_n(\bx_0(\sigma),\sigma)\|=0.
$$
This shows that for $\mu$-a.e. $\sigma\in\varSigma_{\bK}^+$, it is nonchaotic for System (\ref{eq2.1}) in the sense of Definition~\ref{def1.1} because of the lack of the distal property $\mathrm{(\ref{eq1.3b})}$. Similarly, if $\mathcal {S}$ has a Lyapunov exponent $\lambda>0$ at $\mu$ then
$$
\liminf_{n\to+\infty}\|x_n(\bx_0(\sigma),\sigma)\|=+\infty\quad \textrm{for }\mu\textrm{-a.e. }\sigma\in\varSigma_{\bK}^+.
$$
So, there is no the proximal property $\mathrm{(\ref{eq1.3a})}$ for $\mu$-a.e. $\sigma\in\varSigma_{\bK}^+$.
This thus completes the proof of Proposition~\ref{prop2.5}.
\end{proof}

An extreme case is the following proposition.

\begin{prop}\label{prop2.6}
Let $\{S_1,\dotsc,S_K\}$ be irreducible; i.e., there is no a common, invariant, nontrivial and proper subspace of $\mathbb{R}^d$ for all $S_1,\dotsc,S_K$. If at every $\theta$-ergodic probability measure, $\mathcal {S}$ has only the maximal Lyapunov exponent $0$, then every $\sigma\in\varSigma_{\bK}^+$ is nonchaotic for System $(\ref{eq2.1})$ in the sense of Definition~\ref{def1.1}.
\end{prop}

\begin{proof}
Since the hypothesis of the statement implies that $(\ref{eq2.1})$ has the joint spectral radius $1$, from Elsner's reduction theorem~\cite{El} (also see \cite{Dai-JMAA} for a simple proof) it follows that $(\ref{eq2.1})$ is \textit{product bounded}; that is, there exists a constant $0<\beta<+\infty$ such that
\begin{equation}\label{eq2.4}
\|S_{\sigma(n)}\dotsm S_{\sigma(1)}\|\le\beta\quad\forall\sigma\in\varSigma_{\bK}^+\textrm{ and }n\ge1.
\end{equation}
Thus,
\begin{equation*}
\limsup_{n\to+\infty}\|x_n(x_0,\sigma)\|\le\beta\|x_0\|<+\infty\quad \forall\sigma\in\varSigma_{\bK}^+\textrm{ and }x_0\in\mathbb{R}^d.
\end{equation*}
Thus, there is no the distal property $\mathrm{(\ref{eq1.3b})}$ for each $\sigma\in\varSigma_{\bK}^+$.
This proves the proof of Proposition~\ref{prop2.6}.
\end{proof}

This result also shows that our Definition~\ref{def1.1} is essentially different with Definition~\ref{def2.1} of Balde and Jouan.

\section{Chaotic dynamical behaviors}\label{sec3}
Let $S_1,\dotsc,S_K\in\mathrm{GL}(d,\mathbb{R})$ be arbitrarily given. This section will be mainly devoted to proving our main result Theorem~\ref{thm1.2} stated in Section~\ref{sec1}.

For System (\ref{eq1.1}), let $\Lambda$ be the set that consists of the switching laws $\sigma\in\varSigma_{\bK}^+$ such that
\begin{align*}
&\liminf_{n\to+\infty}\|S_{\sigma(n)}\dotsm S_{\sigma(1)}\|=0\\
\intertext{and}
&\limsup_{n\to+\infty}\|S_{\sigma(n)}\dotsm S_{\sigma(1)}\|_{\textrm{co}}=\infty.
\end{align*}
Then each $\sigma\in\Lambda$ is chaotic for System (\ref{eq1.1}) in the sense of Definition~\ref{def1.1}.

To prove our main result Theorem~\ref{thm1.2}, we first need a lemma.

\begin{lem}\label{lem3.1}
Under the context of Theorem~\ref{thm1.2}, $\Lambda$ is a dense subset of $\varSigma_{\bK}^+$.
\end{lem}

\begin{proof}
Let
$(i_1,\dotsc,i_m)\in\bK^m$ and $(j_1,\dotsc,j_n)\in\bK^n$
be such that
\begin{equation*}
\|S_{i_m}\dotsm S_{i_1}\|<1<\|S_{j_n}\dotsm S_{j_1}\|_{\mathrm{co}}.
\end{equation*}
Simply write
\begin{align*}
&\bi=(i_1,\dotsc,i_m),\quad \bi^k=(\stackrel{k\textrm{-folds}}{\overbrace{\bi,\dotsc,\bi}}),\\
&\bj=(j_1,\dotsc,j_n),\quad \bj^k=(\stackrel{k\textrm{-folds}}{\overbrace{\bj,\dotsc,\bj}}),\\
\intertext{and}
&S(\bi)=S_{i_m}\dotsc S_{i_1},\quad S(\bj)=S_{j_n}\dotsc S_{j_1}.
\end{align*}
Let $\bsigma=(\bsigma(1),\bsigma(2),\dotsc)\in\varSigma_{\bK}^+$ and $\epsilon>0$ be arbitrarily given.
Then one can find an integer $N\ge1$ such that for any
$\sigma\in\varSigma_{\bK}^+$, if $\sigma(1)=\bsigma(1),\dotsc,\sigma(N)=\bsigma(N)$, then the distance $d(\bsigma,\sigma)<\epsilon$.
Set
\begin{equation*}
A=S_{\bsigma(N)}\dotsc S_{\bsigma(1)}.
\end{equation*}
Next, we will construct a chaotic switching law $\sigma\in\varSigma_{\bK}^+$ for System (\ref{eq1.1}) with $d(\bsigma,\sigma)<\epsilon$.

Since all the matrices $S_1,\dotsc,S_K$ are nonsingular, we can choose positive integers $\ell_k<L_k$, for $k=1,2,\dotsc$, such that
\begin{align*}
&\|S(\bi)^{\ell_1}A\|<1,\\
&\|S(\bj)^{L_1}S(\bi)^{\ell_1}A\|_{\textrm{co}}>1;\\
&\|S(\bi)^{\ell_2}S(\bj)^{L_1}S(\bi)^{\ell_1}A\|<\frac{1}{2},\\
&\|S(\bj)^{L_2}S(\bi)^{\ell_2}S(\bj)^{L_1}S(\bi)^{\ell_1}A\|_{\textrm{co}}>2;\\
&\vdots\quad\vdots\quad\vdots\\
&\|S(\bi)^{\ell_k}S(\bj)^{L_{k-1}}\dotsm S(\bj)^{L_1}S(\bi)^{\ell_1}A\|<\frac{1}{k},\\
&\|S(\bj)^{L_k}S(\bi)^{\ell_k}S(\bi)^{L_{k-1}}\dotsm S(\bj)^{L_1}S(\bi)^{\ell_1}A\|_{\textrm{co}}>k;\\
&\vdots\quad\vdots\quad\vdots.
\end{align*}
Now it is easy to see that the switching law $\sigma$ defined by
\begin{equation*}
\sigma=(\bsigma(1),\dotsc,\bsigma(N), \bi^{\ell_1},\bj^{L_1}, \bi^{\ell_2},\bj^{L_2},\bi^{\ell_3},\bj^{L_3},\dotsc)
\end{equation*}
is chaotic for System (\ref{eq1.1}) in the sense of Definition~\ref{def1.1} such that $d(\bsigma,\sigma)<\epsilon$.

This completes the proof of Lemma~\ref{lem3.1}.
\end{proof}

Next, we will prove that $\Lambda$ is a $G_\delta$ subset of $\varSigma_{\bK}^+$; that is, $\Lambda$ is the intersection of countable numbers of open sets.

\begin{lem}\label{lem3.2}
For System $(\ref{eq1.1})$, $\Lambda$ is a $G_\delta$ subset of $\varSigma_{\bK}^+$.
\end{lem}

\begin{proof}
For any positive integer $i$, let
\begin{equation*}
\Lambda_i^s=\left\{\sigma\in\varSigma_{\bK}^+\colon\forall n_0\in\mathbb{N}, \exists n>n_0\textrm{ with }\|S_{\sigma(n)}\dotsm S_{\sigma(1)}\|<\frac{1}{i}\right\}.
\end{equation*}
Then
\begin{equation*}
\Lambda_i^s=\bigcap_{n_0=1}^\infty\bigcup_{n>n_0}\left\{\sigma\in\varSigma_{\bK}^+\colon\|S_{\sigma(n)}\dotsm S_{\sigma(1)}\|<\frac{1}{i}\right\}.
\end{equation*}
Since $\left\{\sigma\in\varSigma_{\bK}^+\colon\|S_{\sigma(n)}\dotsm S_{\sigma(1)}\|<\frac{1}{i}\right\}$ is open in $\varSigma_{\bK}^+$ for every $i$, $\Lambda_i^s$ is a $G_\delta$ set. Thus,
\begin{equation*}
\Lambda^s:=\bigcap_{i=1}^\infty\Lambda_i^s
\end{equation*}
is also a $G_\delta$ set.
On the other hand, let
\begin{equation*}
\Lambda_i^u=\left\{\sigma\in\varSigma_{\bK}^+\colon\forall n_0\in\mathbb{N}, \exists n>n_0\textrm{ with }\|S_{\sigma(n)}\dotsm S_{\sigma(1)}\|_{\textrm{co}}>i\right\}.
\end{equation*}
Then
\begin{equation*}
\Lambda_i^u=\bigcap_{n_0=1}^\infty\bigcup_{n>n_0}\left\{\sigma\in\varSigma_{\bK}^+\colon\|S_{\sigma(n)}\dotsm S_{\sigma(1)}\|_{\textrm{co}}>i\right\}.
\end{equation*}
Moreover
\begin{equation*}
\Lambda^u:=\bigcap_{i=1}^\infty\Lambda_i^u
\end{equation*}
is a $G_\delta$ set.
Therefore, $\Lambda=\Lambda^s\cap\Lambda^u$ is a $G_\delta$ subset of $\varSigma_{\bK}^+$.

This completes the proof of Lemma~\ref{lem3.2}.
\end{proof}

Based on Lemmas~\ref{lem3.1} and \ref{lem3.2} we are now ready to finish the proof of Theorem~\ref{thm1.2}.

\begin{proof}[Proof of Theorem~\ref{thm1.2}]
We easily see that $\Lambda$ is a dense $G_\delta$ subset of $\varSigma_{\bK}^+$ from Lemmas~\ref{lem3.1} and \ref{lem3.2}. Since each $\sigma\in\Lambda$ is chaotic for System (\ref{def1.1}), the set of all chaotic laws of System (\ref{eq1.1}) is residual. This proves Theorem~\ref{thm1.2}.
\end{proof}

Let us consider an example.

\begin{example}
Given any two constants $\alpha,\beta$ such that $|\alpha|<1$ and $|\beta|>1$, let
\begin{equation*}
S_1=\alpha\left(\begin{matrix}1&1\\0&1\end{matrix}\right)\quad \textrm{and}\quad S_2=\beta\left(\begin{matrix}1&0\\1&1\end{matrix}\right).
\end{equation*}
Then from Theorem~\ref{thm1.2}, it follows that System (\ref{eq1.1}) generated by $S_1$ and $S_2$ is chaotic in the sense of Definition~\ref{def1.1}.
\end{example}

We now turn to another basic property of chaotic systems.

\begin{defn}
System $(\ref{eq1.1})$ is called \textit{irreducibly product unbounded} if restricted to every nonempty, common and invariant subspace of $\mathbb{R}^d$, it is product unbounded.
\end{defn}

So, if System $(\ref{eq1.1})$ is irreducibly product unbounded then it is product unbounded. But the converse is not necessarily true. For example, for
\begin{equation*}
S=\left(\begin{matrix}1&1\\0&1\end{matrix}\right)
\end{equation*}
it is product unbounded but not irreducibly product unbounded.

We will employ the following simple fact in the next section.

\begin{lem}\label{lem3.5}
If System $(\ref{eq1.1})$ has a chaotic switching law in the sense of Definition~\ref{def1.1}, then it is irreducibly product unbounded.
\end{lem}

\begin{proof}
This follows immediately from the definitions.
\end{proof}
\section{Periodical stability implies nonchaoticity }\label{sec4}
Recall that System $(\ref{eq2.1})$ described as in Section~\ref{sec2} is called, from e.g. \cite{Gur-LAA95, SWMWK, DHX-aut}, \textit{periodically stable} if for any finite-length words $(k_1,\dotsc,k_\pi)\in\bK^\pi, \pi\ge1$, there holds that the spectral radius $\rho_{S_{k_\pi}\dotsm S_{k_1}}^{}$ of $S_{k_\pi}\dotsm S_{k_1}$ is less than $1$. Then a periodically stable system $(\ref{eq2.1})$ does not need to be absolutely stable
from \cite{BM, BTV, Koz07, HMST}; but it is almost surely exponentially stable in terms of ergodic measures, see \cite{DHX-aut} and \cite[Theorem~C$^\prime$]{Dai-JDE}. The following result further shows that a periodically stable system has no chaotic dynamics in our sense of Definition~\ref{def1.1}.

\begin{theorem}\label{thm4.1}
If System $(\ref{eq2.1})$ has the Rota-Strang joint spectral radius
\begin{equation*}
\pmb{\rho}:=\lim_{n\to+\infty}\max_{\sigma\in\varSigma_{\bK}^+}\sqrt[n]{\|S_{\sigma(n)}\dotsm S_{\sigma(1)}\|}=1,
\end{equation*}
then its every switching law is not chaotic in the sense of Definition~\ref{def1.1}.
\end{theorem}

\begin{proof}
According to Definition~\ref{def1.1}, if System (\ref{eq2.1}) is product bounded as in (\ref{eq2.4}), then it does not have any chaotic switching laws.
By contradiction, we let $\bsigma\in\varSigma_{\bK}^+$ be chaotic for System (\ref{eq2.1}) in the sense of Definition~\ref{def1.1}. Then System (\ref{eq2.1}) is irreducibly product unbounded by Lemma~\ref{lem3.5}. From Elsner's reduction theorem \cite{El, Dai-JMAA}, there is no loss of generality in assuming
\begin{equation*}
S_k=\begin{pmatrix}S_k^{1,1}&*\\0&S_k^{2,2}\end{pmatrix},\quad k=1,\dotsc,K,
\end{equation*}
such that
\begin{equation*}
S_k^{1,1}\in\mathbb{R}^{d_1\times d_1}\quad\textrm{and}\quad S_k^{2,2}\in\mathbb{R}^{(d-d_1)\times (d-d_1)},\quad k=1,\dotsc,K,
\end{equation*}
for some integer $1\le d_1<d$. Clearly, the inclusion system based on $\{S_1^{1,1}, \dotsc, S_K^{1,1}\}$ is also periodically stable and moreover, $\pmb{\sigma}$ is a chaotic switching law for it too. Repeating this argument finite times, we can conclude a contradiction to the irreducible product unboundedness.

This completes the proof of Theorem~\ref{thm4.1}.
\end{proof}

\begin{cor}\label{cor4.2}
If System $(\ref{eq2.1})$ is periodically stable, then its every switching law is not chaotic in the sense of Definition~\ref{def1.1}.
\end{cor}

\begin{proof}
This comes from the Berger-Wang spectral formula~\cite{BW} and \cite{Bar} and Theorem~\ref{thm4.1}.
\end{proof}

This shows that a periodically stable inclusion system
is ``simple'' from our viewpoint of chaos dynamics. In fact, the following Lemma~\ref{lem4.3} shows a low dimensional periodically stable system is product bounded.

It is a well-known fact that for System (\ref{eq2.1}), if it holds that $\rho_{S_{\sigma(n)}\dotsm S_{\sigma(1)}}^{}\le1$ for all $\sigma\in\varSigma_{\bK}^+$, then
\begin{equation*}
\max_{\sigma\in\varSigma_{\bK}^+}\|S_{\sigma(n)}\dotsm S_{\sigma(1)}\|=\mathrm{O}(n^{d-1});
\end{equation*}
see, e.g., \cite{BW, Bel}. In the periodically stable case (or equivalently, $\rho_{S_{\sigma(n)}\dotsm S_{\sigma(1)}}^{}<1\;\forall \sigma\in\varSigma_{\bK}^+$ and $n\ge1$), we can get a more subtle estimate as follows.

\begin{lem}\label{lem4.3}
Let System $(\ref{eq2.1})$ be periodically stable with dimension $d\ge2$. Then
\begin{equation*}
\max_{\sigma\in\varSigma_{\bK}^+}\|S_{\sigma(n)}\dotsm S_{\sigma(1)}\|=\mathrm{O}(n^{\lfloor d/2-1\rfloor}).
\end{equation*}
Here $\lfloor x\rfloor$ represents the largest integer which is not greater than
$x$ for any $x\ge0$.

In particular, if $1\le d\le3$ then System $(\ref{eq2.1})$ is product
bounded in $\mathbb{R}^{d\times d}$; if $4\le d\le 5$ then $\|S_{\sigma(n)}\dotsm S_{\sigma(1)}\|$ is at most linearly increasing.
\end{lem}

\begin{proof}
We will prove the statement by induction on the dimension $d$ of System (\ref{eq2.1}). We first notice that if System (\ref{eq2.1}) is periodically stable with dimension $d=1$, then
there exists a constant $0<\gamma<1$
so that
$$\|S_{\sigma(n)}\dotsm S_{\sigma(1)}\|\le\gamma^n$$
for all $\sigma\in\varSigma_{\bK}^+$ and all $n\ge1$.

Let $m\ge2$ be arbitrarily given. Assume that the statement is true for $d<m$. It suffices to claim that the
statement is also true for $d=m$.

Let $d=m$. The periodical stability of System (\ref{eq2.1}) implies that the joint spectral radius $\pmb{\rho}\le
1$. If System (\ref{eq2.1}) is product bounded then we are done. Otherwise, according to Elsner's reduction theorem we can assume that each $S_k$ has the form
\begin{equation*}
S_k=\begin{pmatrix}S_k^{1,1}& B_k\\ 0& S_k^{2,2}\end{pmatrix},\quad k=1,\dotsc,K,
\end{equation*}
where
$\left\{S_1^{1,1},\dotsc,S_K^{1,1}\right\}\subset\mathbb{R}^{d_1\times
d_1}$, $\{B_1,\dotsc,B_K\}\subset\mathbb{R}^{d_1\times(m-d_1)}$, and
$\left\{S_1^{2,2},\dotsc,S_K^{2,2}\right\}\subset\mathbb{R}^{(m-d_1)\times
(m-d_1)}$ for some $1\leq d_1<m$.
Thus, for any $\sigma\in\varSigma_{\bK}^+$ and $n\ge1$
\begin{equation*}
S_{\sigma(n)}\dotsc S_{\sigma(1)}=\begin{pmatrix}S_{\sigma(n)}^{1,1}\dotsm S_{\sigma(1)}^{1,1}&\spadesuit_{\sigma(n)\dotsm\sigma(1)}\\
0& S_{i_1}^{(2,2)}\cdots
S_{\sigma(n)}^{2,2}\dotsm S_{\sigma(1)}^{2,2}\end{pmatrix}
\end{equation*}
where
\begin{equation*}
\spadesuit_{\sigma(n)\dotsm\sigma(1)}=\sum_{j=1}^{n}S_{\sigma(n)}^{1,1}\dotsm S_{\sigma(j+1)}^{1,1}B_{\sigma(j-1)}S_{\sigma(n)}^{2,2}\dotsc S_{\sigma(1)}^{2,2}.
\end{equation*}
We can choose a constant $C_1>0$ such that
\begin{equation*}
\|B_k\|\le C_1\qquad\forall k=1,\dotsc,K.
\end{equation*}
Now we only need to consider the following two cases.

Case I: When $d_1=1$ or $m-d_1=1$, we can obtain either
\begin{equation*}
\|S_k^{1,1}\|\le\gamma<1\quad\textrm{for }1\le k\le K
\end{equation*}
or
\begin{equation*}
\|S_k^{2,2}\|\le\gamma<1\quad\textrm{for }1\le k\le K,
\end{equation*}
for some constant $0<\gamma<1$. Hence we have
\begin{equation*}
\|\spadesuit_{\sigma(n)\dotsm\sigma(1)}\|\leq
\begin{cases}
C_1C& \textrm{if }m=2;\\
C_1Cn^{\lfloor m/2-1\rfloor}(1+\gamma+\dotsm+\gamma^{n-1})& \textrm{if }m>2,
\end{cases}
\end{equation*}
by the induction assumption, for some constant $C>0$ that is independent of the choices of the switching law $\sigma$. Thus the statement holds in this case.

Case II: When $2\leq d_1<m-1$, according to the induction assumption, it
follows that
\begin{equation*}
\|\spadesuit_{\sigma(n)\dotsm\sigma(1)}\|\leq
C_1Cn^{\lfloor d_1/2-1\rfloor}Cn^{\lfloor(m-d_1)/2-1\rfloor}n
\leq C_1C^2n^{[m/2-1]}.
\end{equation*}
Here we have used the following inequality:
\begin{equation*}
\left\lfloor\frac{d_1}{2}-1\right\rfloor+\left\lfloor\frac{m-d_1}{2}-1\right\rfloor+1\leq \left\lfloor\frac{m}{2}-1\right\rfloor,
\end{equation*}
for any $2\leq d_1<m-1$, which implies the desired result.

This completes the proof of Lemma~\ref{lem4.3}.
\end{proof}

This lemma together with Lemma~\ref{lem2.3} implies the following stability result.

\begin{theorem}\label{thm4.4}
Let System $(\ref{eq2.1})$ be periodically stable with dimension $2\le d\le 3$. Then,
\begin{equation*}
\|S_{\sigma(n)}\dotsm S_{\sigma(1)}\|\to0 \quad \textrm{as }n\to+\infty,
\end{equation*}
for every Balde-Jouan nonchaotic switching laws $\sigma\in\varSigma_{\bK}^+$.
\end{theorem}

\begin{proof}
From Lemma~\ref{lem4.3}, it follows that System (\ref{eq2.1}) is product bounded in the sense as in (\ref{eq2.4}). So, we can define a norm $\pmb{\|}\cdot\pmb{\|}$ on $\mathbb{R}^{d\times d}$ such that $\pmb{\|}S_k\pmb{\|}\le1$ for all $k=1,\dotsc,K$.

Let $\sigma\in\varSigma_{\bK}^+$ be an arbitrary nonchaotic switching laws of Balde and Jouan as in Definition~\ref{def2.1}. Then we can choose some $\kappa$ as in Lemma~\ref{lem2.3}. Since $\rho_{S_\kappa}^{}<1$, we can find some $N>1$ such that $\pmb{\|}S_\kappa^N\pmb{\|}<1$. Then the statement comes from Lemma~\ref{lem2.3} and the sub-multiplicity of matrix norm.
\end{proof}

A side consequence of Lemma~\ref{lem4.3} is the following statement.

\begin{prop}\label{prop4.5}
Let System $(\ref{eq2.1})$ be periodically stable with dimension $2\le d\le 3$. Then, there holds at least one of the following two statements.
\begin{enumerate}
\item[$(1)$] $\mathrm{(Finiteness\ of\ spectrum)}$ There is a word $(k_1,\dotsc,k_\pi)\in\bK^\pi$, for some $\pi\ge1$, such that
\begin{equation*}
\pmb{\rho}=\sqrt[\pi]{\rho_{S_{k_\pi}\dotsm S_{k_1}}^{}}.
\end{equation*}

\item[$(2)$] $\mathrm{(Finiteness\ of\ norm)}$ There exists an extremal norm $\|\cdot\|_*$, defined on $\mathbb{R}^{d\times d}$, such that
\begin{equation*}
\pmb{\rho}=\max_{\sigma\in\varSigma_{\bK}^+}\sqrt[n]{\|S_{\sigma(n)}\dotsm S_{\sigma(1)}\|_*}\quad\forall n\ge1.
\end{equation*}
\end{enumerate}
Here $\pmb{\rho}$ is defined as in Theorem~\ref{thm4.1}.
\end{prop}

\begin{proof}
If statement (1) of Proposition~\ref{prop4.5} holds, then we are done. Otherwise, without loss of generality we may assume System
(\ref{eq2.1}) is periodically stable and $\pmb{\rho}=1$. Thus it is product bounded according to Lemma~\ref{lem4.3}.
Then there exists a vector norm $\|\cdot\|_*$ defined on $\mathbb{R}^{d}$, where $2\le d\le 3$, such that $\|S_k\|_*\le 1$ for any $1\le k\le K$. Therefore, one has
$$\sqrt[n]{\|S_{\sigma(n)}\dotsm S_{\sigma(1)}\|_*}\le \pmb{\rho}\quad\forall \sigma\in\varSigma_{\bK}^+\textrm{ and }n\ge1.$$
This implies that
\begin{align*}
\pmb{\rho}&=\inf_{n\ge1}\left\{\max_{\sigma\in\varSigma_{\bK}^+}\sqrt[n]{\|S_{\sigma(n)}\dotsm S_{\sigma(1)}\|_*}\right\}\\
&\le\sup_{n\ge1}\left\{\max_{\sigma\in\varSigma_{\bK}^+}\sqrt[n]{\|S_{\sigma(n)}\dotsm S_{\sigma(1)}\|_*}\right\}\\
&\le\pmb{\rho},
\end{align*}
and the proof of Proposition~\ref{prop4.5} is thus completed.
\end{proof}

We note here that this result cannot be proved by directly reducing the dimension $d$, since an extremal norm of some sub-blocks of System (\ref{eq2.1}) does not need to be an extremal norm for the full dimensional case.

We ends this section with some remarks on Proposition~\ref{prop4.5}.

\begin{remark}
For System (\ref{eq2.1}) in the case of $2\le d\le 3$, if spectral finiteness property does not hold, then there
exists an extremal norm $\|\cdot\|_*$. Conversely, the non-existence of an extremal norm implies that the finiteness property must hold.
\end{remark}

\begin{remark}
Based on \cite{BM, BTV, Koz07, HMST} we can easily see that Proposition~\ref{prop4.5} does not need to hold for the case $d\ge4$. In fact, there are uncountably many values of the real parameters $\alpha,
\beta$ such that for each pair $(\alpha,\beta)$,
$F=\{F_1,F_2\}$ is periodically stable, where
$$
F_1=\alpha\begin{pmatrix}1&1\\ 0&1\end{pmatrix}\quad \textrm{and} \quad
F_2=\beta\begin{pmatrix}1&0\\ 1&1\end{pmatrix};
$$
but there is at least one
switching law $\pmb{\sigma}\in\varSigma_{\bK}^+$ where $\bK=\{1,2\}$ such that
\begin{equation*}
\|F_{\pmb{\sigma}(n)}\cdots F_{\pmb{\sigma}(1)}\|\not\to 0\quad \textrm{as }n\to+\infty.
\end{equation*}
Define
$$
S_1=\begin{pmatrix}F_1&F_1\\ 0&F_1\end{pmatrix}\quad \textrm{and}\quad S_2=\begin{pmatrix}F_2&F_2\\ 0&F_2\end{pmatrix}.
$$
Then for any $\sigma\in\varSigma_{\bK}^+$ and any $n\ge1$, we have
$$
S_{\sigma(n)}\dotsm S_{\sigma(1)}=\begin{pmatrix}F_{\sigma(n)}\dotsm F_{\sigma(1)}&n F_{\sigma(n)}\dotsm F_{\sigma(1)}\\ 0&F_{\sigma(n)}\dotsm F_{\sigma(1)}\end{pmatrix}.
$$
For $\pmb{\sigma}$, we particularly get
$$\limsup_{n\to+\infty}\|S_{\pmb{\sigma}(n)}\dotsm S_{\pmb{\sigma}(1)}\|=+\infty$$
for any norm $\|\cdot\|$ on $\mathbb{R}^{4\times 4}$.
\end{remark}
\section{Concluding remarks}
In this paper, we have introduced the dynamical concept---chaotic switching laws---for a discrete-time linear inclusion dynamical system that is induced by finitely many nonsingular square matrices. We have proven that if the inclusion system has a stable word and meanwhile an expanding word, then its chaotic switching laws form a residual subset of its all possible switching laws. Therefore in this case, the dynamical behavior of this inclusion system is unpredictable.

\end{document}